\documentclass[aps,pra,twocolumn,superscriptaddress,longbibliography,nofootinbib]{revtex4-1}
\usepackage[normalem]{ulem}
\usepackage{float}
\usepackage{graphicx}  
\usepackage{dcolumn}          
\usepackage{amssymb}
\usepackage{appendix}
\usepackage{physics}   
\usepackage{mathtools}
\usepackage{esvect}
\usepackage{wrapfig}
\usepackage{amsthm}
\usepackage{verbatim}
\usepackage{bbm}
\usepackage{placeins}
\usepackage[colorlinks=true,linkcolor=blue,citecolor=red,plainpages=false,pdfpagelabels]{hyperref}
\usepackage{bm}

\usepackage[mathscr]{euscript}

\def\Tr{\operatorname{Tr}}

\def\({\left(}
\def\){\right)}
\def\[{\left[}
\def\]{\right]}

\newtheorem{lemma}{Lemma}
\newtheorem{observation}{Observation}

\newtheorem{proposition}{Proposition}

\def\>{\rangle}
\def\<{\langle}

\begin{document}

\widetext

\title{Quantum steering under constrained free-will}

\author{Abhishek Sadhu} \email{abhisheks@rri.res.in}
\affiliation{Centre for Quantum Science and Technology (CQST), Center for Security, Theory and Algorithmic Research (CSTAR), International Institute of Information Technology, Hyderabad, Gachibowli, Telangana 500032, India}
\affiliation{Raman Research Institute, Bengaluru, Karnataka 560080, India}

\author{Siddhartha Das}\email{das.seed@iiit.ac.in}
\affiliation{Centre for Quantum Science and Technology (CQST), Center for Security, Theory and Algorithmic Research (CSTAR), International Institute of Information Technology, Hyderabad, Gachibowli, Telangana 500032, India}

\date{\today}
\begin{abstract}
Quantum steering is a kind of bipartite quantum correlations where one party's measurement remotely alters the state of another party. In an adversarial scenario, there could be a hidden variable introducing a bias in the choice of measurement settings of the parties. However, observers without access to the hidden variable are unaware of this bias. The main focus of this work is to analyze quantum steering without assuming that the parties freely choose their measurement settings. For this, we introduce the measurement-dependent (MD-)steering scenario where the measurement settings chosen by the parties are biased by an adversary. In such a scenario, we present a class of inequalities to test for MD-steerable correlations. Further, we discuss the implications of violating such inequalities in certifying randomness from quantum extremal behaviors. We also assume that an adversary might prepare an assemblage as a mixture of MD-steerable and MD-unsteerable assemblages and provide a bound on the measurement dependence for the observed correlation to remain MD-steerable. 
\end{abstract}

\maketitle

\section{Introduction}
Quantum correlations are of wide interest from both fundamental as well as technological aspects~\cite{Schrodinger35,dowling2003quantum,CSW14,DBWH21,SSHD23}. Quantum steering~\cite{UCNG20,CS17} is a kind of quantum correlations between outcomes of measurements applied on one part of an entangled state and the post-measurement state left with the other party. Operationally, quantum steering can be seen as a test for entanglement where one of the party's measurements is uncharacterised~\cite{WJD07}. Quantum steering is a resource for different information processing tasks such as quantum key distribution~\cite{BCW+12,GHD+15}, randomness generation~\cite{LTB+14,PCS+15,SC+18,GCH+19,sarkar23}, quantum teleportation~\cite{Reid13,HZA15}, and quantum metrology~\cite{YFG21,LLM+23}. 

Schr{\"o}dinger~\cite{Schrodinger35} introduced the concept of quantum steering in response to the 1935 paper by Einstein, Podolsky, and Rosen (EPR)~\cite{EPR35}. EPR argued that in a correlated bipartite state, measuring an observable on one particle enables predicting the outcome of the same measurement on the other. Schr{\"o}dinger noted that by performing a measurement, the first party can \textit{steer} the state of the other party into an eigenstate of the measured observable. 

The study of EPR steering received considerable interest after Wiseman \textit{et al.,}~\cite{WJD07,JWD07} formalized steering for mixed states and provided an operational interpretation via the task of entanglement verification. As there is no assumption of trusted devices at the site of the steering party, it is often termed as one-sided device-independent scenario (1SDI) for entanglement verification. 
\begin{figure}[h]
    \centering
    \includegraphics[scale=0.24]{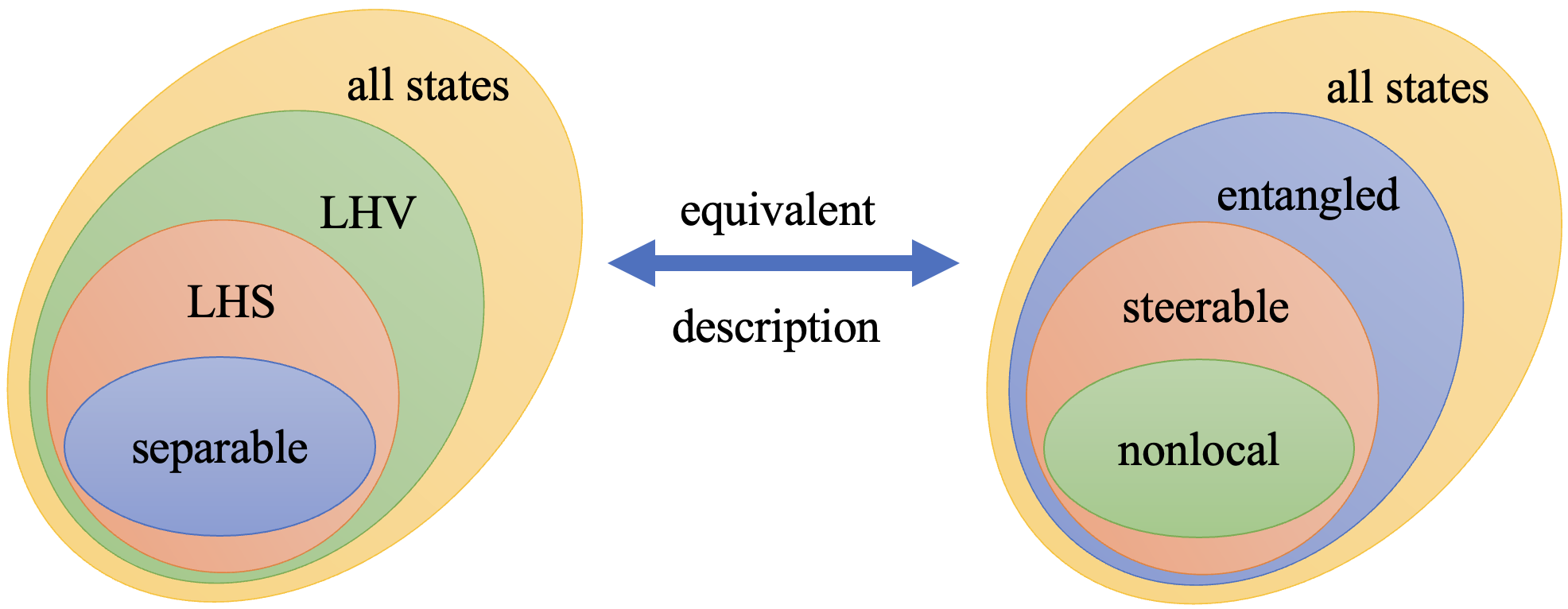}
    \caption{The relation between nonlocality, steering and entanglement: the set of states that has an LHV model, i.e., does not violate any Bell inequality form a convex subset of the set of all states. The set of states that has a LHS model, i.e., are steerable form a convex subset of the LHV states. Also, the separable states form a convex subset of the LHS states (see \cite{UCNG20} for details).}
    \label{fig:hierarchy}
\end{figure}
EPR steering is a form of quantum correlation that lies between entanglement~\cite{HHHH09} and Bell nonlocality~\cite{BCP+14,HSD15} (see Fig.~\ref{fig:hierarchy}). Steering is a strictly stronger form of correlation than entanglement and is strictly weaker than the violation of a Bell inequality~\cite{QVC+15}. Unlike entanglement and Bell nonlocality, steering is a directional form of correlation. The steerability from subsystem A to subsystem B may differ from the steerability in the opposite direction~\cite{WJD+08,MFO+10,HES+12}. The aspect of one-way steering has also been studied in the context of multipartite systems~\cite{HR13,JSU+21,CSA+15}. 

The free-will assumption\footnote{It should be noted that the free-will assumption is also known as measurement independence assumption. This refers to the correlation between the choice of measurement settings of the user that can affect the observed statistics.} in tests of local realistic hidden variable (LRHV) theories states that the users can choose the measurement settings freely or use uncorrelated random number generators. The authors of \cite{big2018challenging} have performed a test of LRHV theories using human-generated randomness in an attempt to close this loophole. To the best of our knowledge, the free-will assumption in a test of LRHV theories was first relaxed in \cite{Hall10}, where the authors use a distance-measure-based quantifier for measurement dependence. In another approach \cite{putz2014arbitrarily,putz2016measurement}, measurement dependence was quantified by bounding the probability of choosing the measurement settings to be in a given range. Following these two approaches, tests for LRHV models have been constructed~\cite{putz2014arbitrarily,ZRL+23,SD23}. However the consideration of measurement dependence in tests of other forms of quantum correlations such as quantum steering is still lacking. 

The main focus of this work is to study the implications of the measurement dependence assumption in the bipartite quantum steering scenario. We introduce measurement-dependent local hidden state (MD-LHS) models and present MD-quantum steering as the generation of ensembles that are not explainable by MD-LHS models. We show that the joint probability distribution obtained from performing local measurements on MD-LHS assemblage has an MD-LHV model. This implies that MD-unsteerable states are also MD-local, making it unsuitable for device-independent information processing tasks. We introduce measurement weight as a quantifier of the distinguishability of the measurement settings of one of the parties. Then, assuming that an observed assemblage of one of the parties may be expressed as a convex combination of steerable and MD-unsteerable assemblages, we present a bound on the measurement-dependent weight such that the observed assemblage is steerable. 

We also introduce an inequality to test for MD-LHV-LHS models and present the maximum achievable violation for bipartite qubit states. We also present the amount of violation of the inequality for PR box correlations. Finally, we discuss the implications and tradeoffs between the amount of violation of the MD-LHV-LHS inequality and measurement dependence parameter for the quantum behaviors that certify different amounts of DI-randomness.

The structure of the paper is as follows. In Sec.~\ref{sec:MD-steering-scenario}, we introduce the framework of quantum steering and measurement dependence, a constraint that limits the free-will of the user. We show that the measurement-dependent unsteerable states are also measurement-dependent local. We introduce measurement-dependent weight in Sec.~\ref{sec:MD-steering-weight} to quantify one-sided measurement dependence. We provide a bound on the measurement-dependent weight to ensure that the observed assemblage is steerable. In Sec.~\ref{sec:MD-LHV-LHS}, we present an inequality to test for MD-LHV-LHS models. We present the amount of violation of the MD-LHV-LHS inequality achievable by quantum states and PR box correlations. We present the amount of violation of the MD-LHV-LHS inequality as a function of the measurement dependence parameter for certifying different amounts of randomness in Sec.~\ref{sec:Randomness}. We provide concluding remarks in Sec.~\ref{sec:discussion}.

\section{The MD-steering Scenario} \label{sec:MD-steering-scenario}
Consider the Bell scenario where two parties, Alice and Bob, perform measurements $\{x\}_x$ and $\{y\}_y$ respectively on their share of a bipartite quantum state $\rho_{AB}$ (see Fig.~\ref{fig:steeringScenario}).
\begin{figure}[h]
    \centering
    \includegraphics[scale=0.25]{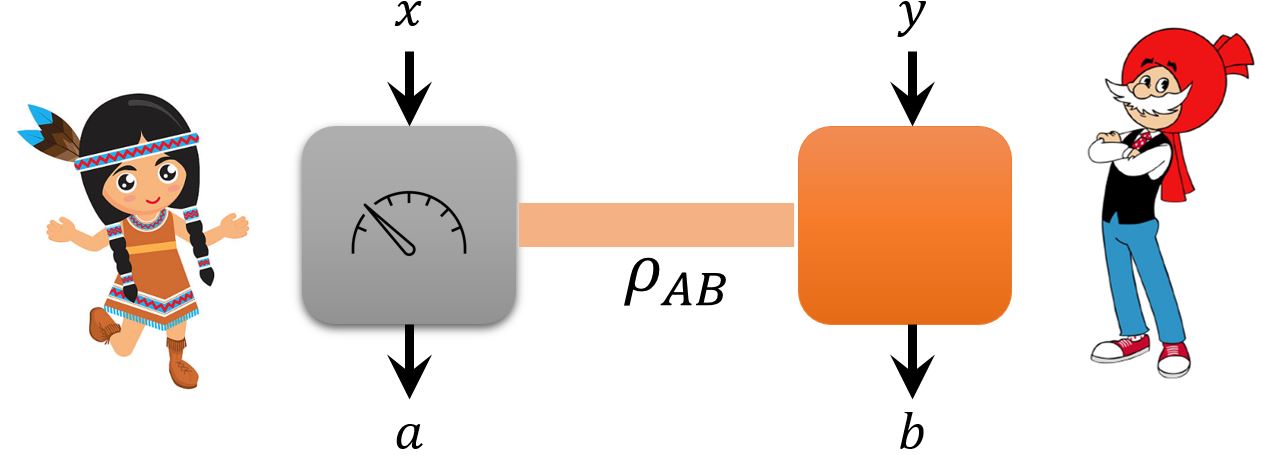}
    \caption{The Bell scenario with two parties, namely Alice and Bob, performing measurements $x$ and $y$ respectively on their share of a bipartite quantum state $\rho_{AB}$. After measuring $x$ and receiving outcome $a$, the state of Bob is $\rho_{a|x}$ with a probability $p(a|x)$. We assume there is no characterisation of Alice's measurement while Bob has full control over his measurement settings and can access the conditional state $\rho_{a|x}$.}
    \label{fig:steeringScenario}
\end{figure}
After Alice has measured $x$ and received outcome $a$, the state of Bob is reduced to $\rho_{a|x}$ with a probability of $p(a|x)$. We consider the scenario where the device of Alice is not characterised while Bob has full control over his measurement settings and can access the conditional state $\rho_{a|x}$. This scenario is characterised by a set of unnormalised quantum states $\{\sigma_{a|x}\}_{a,x}$, where $\sigma_{a|x} = p(a|x) \rho_{a|x}$ is called an assemblage and is given by
\begin{equation}
    \sigma_{a|x} = \Tr_A [(\Lambda_a^x \otimes \mathbbm{1})\rho_{AB}].
\end{equation}
\textit{Measurement-dependent (MD-)quantum steering} can be described as the generation of ensembles that are not explainable by \textit{measurement-dependent local hidden state (MD-LHS)} models.

In an MD-LHS model, a source sends a classical message $\lambda$ to Alice, distributed according to $p(\lambda|x)$, while Bob receives the corresponding quantum state $\rho_{\lambda|x}$. On measuring $x$ on her subsystem, Alice obtains outcome $a$ with probability $p(a|x\lambda)$. After Alice has performed her measurement, Bob observes the assemblage element
\begin{equation} \label{eq:steeringDef}
    \sigma_{a|x}^{\text{MD-LHS}} = \int d\lambda p(\lambda|x) p(a|x\lambda) \rho_{\lambda|x}.
\end{equation}
The assemblage elements of Bob reduce to the expression for measurement independence scenario with $p(\lambda|x) = p(\lambda)$ and $\rho_{\lambda|x} = \rho_{\lambda}$. The assemblage elements have the following properties: (a) \textit{Positivity}: $\sigma_{a|x}^{\text{MD-LHS}} \geq 0$ and (b) \textit{Normalization}: $\sum_a \Tr \sigma_{a|x}^{\text{MD-LHS}} = 1~\forall x$. The assemblages that cannot be described in the form of Eq.~\eqref{eq:steeringDef} are called \textit{MD-steerable}. 

\begin{lemma} \label{prop:steerNonlocal}
    The probability distribution obtained from local measurements on an MD-LHS assemblage has an MD-LHV model.
\end{lemma}
\begin{proof}
The probability distribution $\textbf{P} = \{p(ab|xy)\}_{a,b,x,y}$ obtained after Bob has performed a local measurement on the MD-LHS assemblage $\sigma_{a|x}^{\text{MD-LHS}}$ is given by
\begin{eqnarray}
p(ab|xy) &&= \Tr[\Lambda_b^y \sigma_{a|x}^{\text{MD-LHS}}] \\
&&= \Tr[\Lambda_b^y \int d\lambda p(\lambda|x) p(a|x\lambda) \rho_{\lambda|x}] \\
&&= \int d\lambda p(\lambda|x) p(a|x\lambda) \Tr[\Lambda_b^y \rho_{\lambda|x}] \\
&&= \int d\lambda p(\lambda|x) p(a|x\lambda) p(b|y\lambda)_{\rho_{\lambda|x}}, \label{eq:steerNonlocal}
\end{eqnarray}
where $p(\lambda|x)$ denotes the one-sided measurement dependence of Alice. It follows from Eq.~\eqref{eq:steerNonlocal} that probability distribution $\textbf{P}$ has an MD-LHV model. 
\end{proof}
The behavior $\textbf{P}$ which can be expressed in the form of Eq.~\eqref{eq:steerNonlocal} has an \textit{MD-LHV-LHS model}. It follows from Lemma~\ref{prop:steerNonlocal} that MD-unsteerable states are also MD-local and cannot be used for device-independent information processing tasks. 

\section{Measurement dependence and steering weight} \label{sec:MD-steering-weight}
In this section, we introduce \textit{measurement-dependent weight} to quantify the measurement dependence of one of the parties in the bipartite steering scenario. We also provide a bound on the measurement-dependent weight to ensure that the observed behavior is MD-steerable. 

Consider the scenario where Alice is trying to prepare Bob's assemblage $\{\sigma_{a|x}\}_{a,x}$. She can prepare as frequently as possible an MD-unsteerable assemblage $\{\sigma_{a|x}^{\text{MD-LHS}}\}_{a,x}$, and also sometimes a steerable assemblage $\{\gamma_{a|x}^{\text{st}}\}_{a,x}$ such that on average the assemblage is steerable. Taking motivation from~\cite{skrzypczyk2014quantifying}, we can decompose the elements of the assemblage created by Alice as 
\begin{equation} \label{eq:steringMixture}
    \sigma_{a|x} = (1-\eta^{a|x})~\gamma_{a|x}^{\text{st}} + \eta^{a|x}~\sigma_{a|x}^{\text{MD-LHS}},
\end{equation}
where $a \in \{+1,-1\}$ and $x \in \{x_1,x_2\}$. The parameter $\eta^{a|x}\in(0,1)$\footnote{It follows from Eq.~\eqref{eq:steringMixture} that for $\eta^{a|x} = 1$ we have no contribution from the steerable assemblage $\gamma_{a|x}^{\text{st}}$ and $\sigma_{a|x} = \sigma_{a|x}^{\text{MD-LHS}}$, while for $\eta^{a|x} = 0$ we have and $\sigma_{a|x} = \gamma_{a|x}^{\text{st}}$.} quantifies the weight of $\sigma_{a|x}^{\text{MD-LHS}}$ in the mixture. 

Measurement independence of Alice implies no extra information about the hidden variable $\lambda$ can be obtained from the knowledge of her choice of measurement settings\footnote{Measurement independence of Alice imposes the condition: $p(\lambda|x_1) = p(\lambda|x_2) = p(\lambda)$, which is equivalent to saying \begin{eqnarray}
    p(x|\lambda) = \frac{p(\lambda|x)p(x)}{p(\lambda)} = p(x).
\end{eqnarray}
In other words, Alice has complete freedom in choosing her measurement setting.}. To quantify distinguishability between $p(\lambda|x_1)$ and $p(\lambda|x_2)$, we introduce the \textit{measurement-dependent weight} $\mathscr{W}$ for Alice as 
\begin{eqnarray}
    \mathscr{W} = \int d\lambda \bigg( p(\lambda|x_1) - \frac{\eta^{-|x_2}}{\eta^{-|x_1}} p(\lambda|x_2) \bigg). \label{eq:md-para}
\end{eqnarray}
If Alice has complete measurement independence, we have $\mathscr{W} = 1 - (\eta^{-|x_2}/\eta^{-|x_1})$. 

Following~\cite{putz2014arbitrarily,putz2016measurement}, we constrain the correlation between the choice of the measurement setting of Alice and the hidden variable $\lambda$ by imposing lower and upper bounds on the conditional probability distribution 
\begin{eqnarray}
    l \leq p(x|\lambda) \leq 1-l~~ \text{with} ~~ l \in [0,0.5]. \label{eq:MD_para_Alice}     
\end{eqnarray}
Consider the case of $l = 0.5$. It follows that $p(x|\lambda) = p(x) = 0.5$, which implies Alice has measurement independence. Other values of $l$ represent bias in Alice's choice of measurement settings. We next present values of the measurement-dependent weight $\mathscr{W}$ for the limiting values of $p(x|\lambda)$.
\begin{observation}
The value of the measurement-dependent weight $\mathscr{W}$ for limiting values of $p(x|\lambda)$ (see Eq.~\eqref{eq:MD_para_Alice}) is given by 
\begin{align}
    &&\mathscr{W} \coloneqq \begin{cases}
        ~\frac{l}{p(x_1)} - \frac{\eta^{-|x_2}(1-l)}{\eta^{-|x_1} (1-p(x_1))} & \text{if $p(x_1|\lambda) = l$}, \\
        \frac{1-l}{p(x_1)} - \frac{l~\eta^{-|x_2} }{\eta^{-|x_1} (1-p(x_1))} & \text{if $p(x_1|\lambda) = 1-l$}. 
    \end{cases}
\end{align}
\end{observation}
Next, we establish an upper bound on $\mathscr{W}$ that depends on the weight of the MD-LHS assemblage $\sigma_{a|x}^{\text{MD-LHS}}$ such that $\sigma_{a|x}$ is MD-steerable (see Appendix~\ref{app:steering-weight} for proof).
\begin{proposition} \label{prop:weight}
The assemblage $\sigma_{a|x}$ in Eq.~\eqref{eq:steringMixture} is MD-steerable if the measurement-dependent weight $\mathscr{W}$ is bounded by 
\begin{eqnarray}
    &&\mathscr{W} \leq \frac{1}{\eta^{-|x_1}} \bigg[ p(+|x_2)(\eta^{+|x_2} + \eta^{-|x_2}) \nonumber \\
    &&\qquad \qquad \qquad \qquad - p(+|x_1)(\eta^{+|x_1} + \eta^{-|x_1})\bigg]. \label{eq:weight}
\end{eqnarray}
In the limiting case of $\eta^{a|x} = \eta~\forall a,x$, we have $\mathscr{W} \leq 2 [p(+|x_2) - p(+|x_1)]$.
\end{proposition}

\section{Test for MD-LHV-LHS models} \label{sec:MD-LHV-LHS}
In this section, we present an inequality to test for MD-LHV-LHS  models in the simplest case of two parties with binary inputs and outputs on each side. We show that there exists a class of bipartite qubit states that violate the MD-LHV-LHS inequality. We present the amount of violation of the MD-LHV-LHS inequality for different classes of nonlocal behaviors. 

The behavior $\textbf{P} = \{p(ab|xy)\}_{a,b,x,y}$ obtained from the joint measurement on the bipartite state $\rho_{AB}$ shared between Alice and Bob has an MD-LHV-LHS model (see Lemma~\ref{prop:steerNonlocal}) if the joint probability distribution can be written as
\begin{equation} \label{eq:steeringInequality1}
    p(ab|xy) = \sum_{\lambda} p(\lambda|xy) p(a|x\lambda) p(b|y\lambda)_{\rho_{\lambda|x}},
\end{equation}
where the hidden variable $\lambda$ determines a local quantum state $\rho_{\lambda|x}$ for Bob and $p(b|y\lambda)_{\rho_{\lambda|x}}$ is the probability of outcome $b$ if $y$ is measured on $\rho_\lambda$. 

We now present an inequality whose validation certifies the existence of an MD-LHV-LHS model for $\textbf{P}$ (see Appendix~\ref{app:MD-LHV-LHS} for proof).
\begin{proposition}
    The behavior obtained from joint-dichotomic measurements on the bipartite state shared by Alice and Bob has an MD-LHV-LHS model if
\begin{equation}
    I_{\alpha_1}^{\alpha_2}\coloneqq\sqrt{\alpha_1} + \sqrt{\alpha_2} \leq 4p(1-p),\label{eq:steeringInequalityBob1}
\end{equation}
where $p \in [0,0.5]$ is the measurement dependence parameter of Alice and $\alpha_1$ and $\alpha_2$ are given by 
\begin{eqnarray}
    &&\alpha_1 = (p\langle x_1y_1\rangle + (1-p) \langle x_2y_1\rangle)^2 + (p\langle x_1y_2\rangle \nonumber \\
    && \qquad \qquad \qquad \qquad \qquad  + (1-p) \langle x_2y_2\rangle)^2, \\
    &&\alpha_2 = (p\langle x_1y_1\rangle - (1-p) \langle x_2y_1\rangle)^2 + (p\langle x_1y_2\rangle \nonumber \\
    &&\qquad \qquad \qquad  \qquad \qquad - (1-p) \langle x_2y_2\rangle)^2. 
\end{eqnarray}
\end{proposition}
For MD-LHV-LHS models, we observe from Eq.~\eqref{eq:steeringInequalityBob1} that $I_{\alpha_1}^{\alpha_2}$ is upper bounded by $I_{\text{loc}} \coloneqq 4p(1-p)$. In the limiting case of $p = 0.5$, Eq.~\eqref{eq:steeringInequalityBob1} reduces to that obtained in \cite{girdhar2016all} and $I_{\text{loc}} = 1$. Also, in the limit of $p \rightarrow 0$ we observe that $I_{\text{loc}} \rightarrow 0$. We plot the variation of $I_{\text{loc}}$ as a function of $p$ in Fig.~\ref{fig:critp}. Next, we observe the amount of violation of the MD-LHV-LHS inequality for different classes of nonlocal behaviors. Next, we show that there exist quantum states that violate the MD-LHV-LHS inequality. 
\begin{proposition} \label{prop:quantumStates}
    There exist quantum states that violate the MD-LHV-LHS inequality.
\end{proposition}
\begin{proof}
Let Alice and Bob share a two-qubit entangled state $\cos \theta \ket{00} - \sin \theta \ket{11}$ and perform measurements
\begin{eqnarray}
    &&x_1 = \hat{\bm{n}}_1 . \bm{\sigma};~~~  x_2 = \hat{\bm{n}}_2 . \bm{\sigma}, \nonumber \\
    &&y_1 = \hat{\bm{m}}_1 . \bm{\sigma};~~~  y_2 = \hat{\bm{m}}_2 . \bm{\sigma}, \nonumber
\end{eqnarray}    
where $\hat{\bm{m}}_1,\hat{\bm{m}}_2,\hat{\bm{n}}_1,\hat{\bm{n}}_2$ are the measurement directions of Alice and Bob and $\bm{\sigma} = \sigma_x\hat{x}+\sigma_y\hat{y}+\sigma_z\hat{z}$, with $\sigma_j$ being Pauli matrices. By varying over the state parameter $\theta$ and measurement directions of Alice and Bob, we obtain the maximum achievable value of $I_{\alpha_1}^{\alpha_2}$ for a given value of $p$ and plot it in Fig.~\ref{fig:critp}. 
\end{proof}
There exists a class of behaviors called the no-signalling behaviors which are constrained only by the no-signalling conditions~\cite{cirelson80,PR94}. The no-signalling conditions imply that the local marginal probabilities of one of the parties are independent of the measurement settings of the other party. An important class of no-signalling behaviors is the set of PR box correlations~\cite{PR94}. In the bipartite two-input-two-output scenario, they form the nonlocal vertices of the no-signalling polytope~\cite{BCP+14}.  
These correlations return the algebraic maximum value of $4$ for the CHSH expression~\cite{CHSH69} and cannot be obtained from performing measurements on quantum states. 
The value of the MD-LHV-LHS operator $I_{\alpha_1}^{\alpha_2}$ for the PR box correlations is given by
\begin{eqnarray}
    I_{\text{PR-box}} =  2 \sqrt{2 - 4p (1-p)},
\end{eqnarray}
where $p$ is the measurement dependence parameter for Alice. 
In the limiting case of $p = 0.5$, we have $I_{\text{PR}} = 2$ while for $p \rightarrow 0$, we have $I_{\text{PR}} \rightarrow 2\sqrt{2}$. We plot the variation of $I_{\text{PR}}$ as a function of $p$ in Fig.~\ref{fig:critp}. 
\begin{figure}
    \centering
    \includegraphics[scale=0.6]{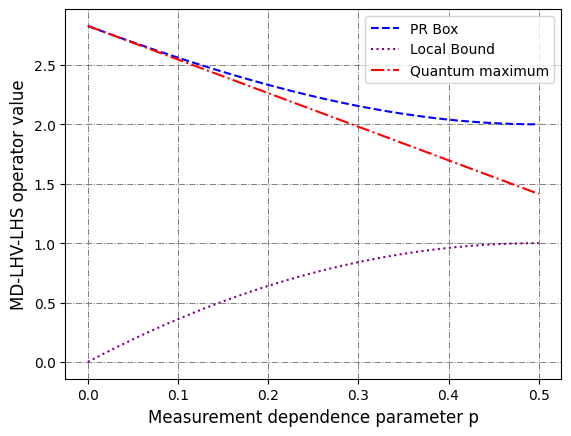}
    \caption{In this figure, we plot the maximum achievable value of $I_{\alpha_1}^{\alpha_2}$ as a function of the measurement dependence parameter $p$ of Alice for (a) the MD-LHV-LHS models (in yellow) (b) quantum states of the form $\cos \theta \ket{00} - \sin \theta \ket{11}$ (in purple). We also plot the value of $I_{\alpha_1}^{\alpha_2}$ as a function of $p$ for PR box correlations (in blue).}
    \label{fig:critp}
\end{figure}

Another important class of nonlocal behaviors is the set of quantum extremal correlations. Such correlations maximally violate a Bell inequality and are useful for device-independent information processing tasks. We present the amount of violation of $I_{\alpha_1}^{\alpha_2}$ required to ensure there is no MD-LHV-LHS model for a class of quantum extremal behaviors.
\begin{observation}
Consider the behavior $\mathbf{P}_1$ obtained from the state $\rho_{AB} = \op{\psi_1}$ where $\ket{\psi_1} = \frac{1}{\sqrt{2}} (\ket{00} +  \ket{11})$ and the choice of measurement setting 
\begin{eqnarray}
    &&x_1 = \sigma_z;~~~  x_2 = -\sin \delta~\sigma_z + \cos \delta~\sigma_x, \nonumber \\
    &&y_1 = \sigma_x;~~~  y_2 = \cos \delta~\sigma_z - \sin \delta~\sigma_x, \nonumber
\end{eqnarray}   
where $0 < \delta \leq \pi/6$. The value of the CHSH operator for the behavior $\mathbf{P}_1$ is given by $\text{CHSH}_\delta = 2 \cos \delta (1 + \sin \delta)$~\cite{WBC22} and $\mathbf{P}_1$ achieves the quantum bound of the tilted Bell expression
\begin{eqnarray}
    I_\delta \coloneqq \langle A_0B_0 \rangle + \frac{1}{\sin \delta}(\langle A_0B_1 \rangle + \langle A_1B_0 \rangle) - \frac{1}{\cos 2\delta} \langle A_1B_1 \rangle. \nonumber \\
\end{eqnarray}
Hence, $\mathbf{P}_1$ is located on the quantum boundary. 
The value of the MD-LHV-LHS operator $I_{\alpha_1}^{\alpha_2}$ for $\mathbf{P}_1$ is given by 
\begin{eqnarray}
    I_1 = &&\sqrt{\cos ^2(\delta ) \left((2 (p-1) \sin (\delta )+p)^2+(p-1)^2\right)} \nonumber \\ 
    && +\sqrt{\cos ^2(\delta ) \left((p-2 (p-1) \sin (\delta ))^2+(p-1)^2\right)}. \nonumber \\
\end{eqnarray}
We denote the amount of violation of the MD-LHV-LHS inequality by the behavior $\mathbf{P}_1$ as $\Delta_1 = I_1 - \eta_p$ where $\eta_p = 4 p (1-p)$ is the local bound. The variation of amount of $\Delta_1$ as a function of the measurement dependence parameter $p$ for different quantum extremal behaviors is plotted in Fig.~\ref{fig:randomess}. 
We observe high values of measurement dependence (denoted by low values of $p$) require a higher violation of the MD-LHV-LHS inequality.
\begin{figure}
    \centering
    \includegraphics[scale=0.6]{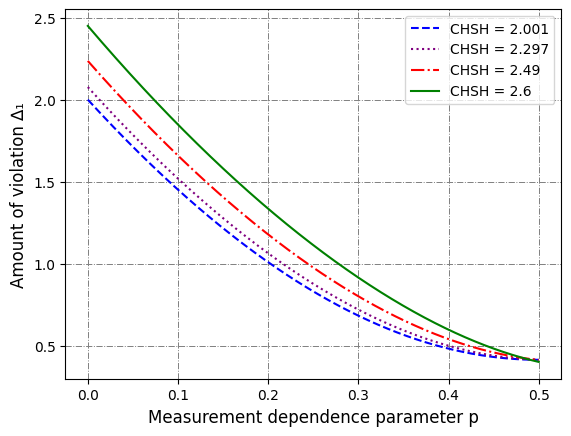}
    \caption{The amount of violation ($\Delta_1$) of the MD-LHV-LHS inequality as a function of the measurement dependence parameter $p$ for different quantum extremal behaviors.}
    \label{fig:randomess}
\end{figure}
\end{observation}

\section{Application: Randomness Extraction} \label{sec:Randomness}
Randomness acts as a primitive resource for various desirable information processing. Also, the generation and certification of private randomness also constitute a major goal of cryptographic protocols. In this section, we present the amount of violation of the MD-LHV-LHS inequalty for quantum extremal behaviors providing different amounts of randomness. 

In the device-independent (DI) scenario, Eve may have supplied the devices of Alice and Bob and hold a purifying system $E$ such that the post-measurement systems of Alice, Bob, and Eve are correlated. We define the joint system of Alice-Bob-Eve using the density operator $\rho_{ABE}$ such that the state of Alice and Bob is given by $\rho_{AB} = \Tr_E [\rho_{ABE}]$. After Alice and Bob perform measurements $x$ and $y$ on their respective sub-systems, we obtain the classical-quantum state $\Tilde{\rho}_{ABE} = \sum_{ab} \op{ab}_{AB} \otimes \rho_{E}^{abxy}$, where $\rho_{E}^{abxy} = \Tr_{AB} [(\Lambda_a^x \otimes \Lambda_b^y \otimes \mathbbm{1}_E) \rho_{ABE}]$ is the state of Eve conditioned on the joint measurement outcomes of Alice and Bob. We obtain the conditional joint probability distribution of the measurement outputs conditioned on the inputs as $p(ab|xy) = \Tr[\rho_{E}^{abxy}]$. 

The amount of DI global randomness for spot-checking DI randomness generation protocols is obtained from the conditional von Neumann entropy $H(ab|x_1,y_1,E)$ evaluated for the state $\Tilde{\rho}_{ABE}$ when Alice and Bob choose measurement settings $x_1$ and $y_1$ respectively~\cite{BRC23}. The asymptotic rate of randomness generation $r$ is the lower bound on the conditional von Neumann entropy for all states and measurement settings that are compatible with the observed behavior $\mathbf{P}_{\text{obs}}$, or a linear function $f$ of $\mathbf{P}_{\text{obs}}$. The asymptotic rate of randomness generation $r$ in bits per round is given by
\begin{eqnarray} \label{eq:randomnessRate}
    r = \inf_{\substack{\rho_{ABE}, \\ \{\Lambda_a^x\}, \{\Lambda_b^y\} \\ \text{compatible with } f(\mathbf{P}_{\text{obs}})}} H(ab|x_1,y_1,E)_{\Tilde{\rho}_{ABE}},
\end{eqnarray}
where $f(\mathbf{P}_{\text{obs}})$ is the Bell value. 

Consider the behavior $\mathbf{P}_2$ obtained from the state $\rho_{AB} = \op{\psi_1}$ where $\ket{\psi_1} = \frac{1}{\sqrt{2}} (\ket{00} +  \ket{11})$ and the choice of measurement settings
\begin{eqnarray}
    &&\Tilde{x}_1 = \sigma_z, ~~~~\Tilde{y}_1 = \sin 3\gamma~\sigma_z + \cos 3\gamma~\sigma_x, \nonumber \\
    && \Tilde{x}_2 = \cos(\frac{2\pi}{3} - 2\gamma) \sigma_z + \sin(\frac{2\pi}{3} - 2\gamma) \sigma_x, \nonumber \\
    && \Tilde{y}_2 = \cos(\frac{\pi}{6} + \gamma) \sigma_z - \sin(\frac{\pi}{6} + \gamma) \sigma_x,
\end{eqnarray}
where $0 \leq \gamma \leq \pi/12$. The amount of randomness obtained from $\mathbf{P}_2$ is given by~\cite{WBC22} 
\begin{eqnarray}
    r = 1 + H_{\text{b}}\bigg[\frac{1}{2} + \frac{s}{2} - \frac{3}{\sqrt{2}}\cos(\frac{1}{3} \arccos\bigg[-\frac{s}{2\sqrt{2}}\bigg])\bigg], \nonumber \\
\end{eqnarray}
where $s = \sin 3\gamma + 3 \cos(\gamma + \pi/6)$ and $H_\text{b}[.]$ denotes the binary entropy.

We denote the amount of violation of the MD-LHV-LHS inequality by $\mathbf{P}_2$ as $\Delta_2 = I_2 - \eta_p$ where $I_2$ is the value of the operator $I_{\alpha_1}^{\alpha_2}$ for the behavior $\mathbf{P}_2$ and $\eta_p = 4p(1-p)$ is the local bound. The variation of $\Delta_2$ as a function of the measurement dependence parameter p for different amounts of randomness $r$ is plotted in Fig.~\ref{fig:randomess2}. We observe that high values of measurement dependence (denoted by low values of $p$) and also high values of $r$ require a higher violation of the MD-LHV-LHS inequality.

\begin{figure}
    \centering
    \includegraphics[scale=0.55]{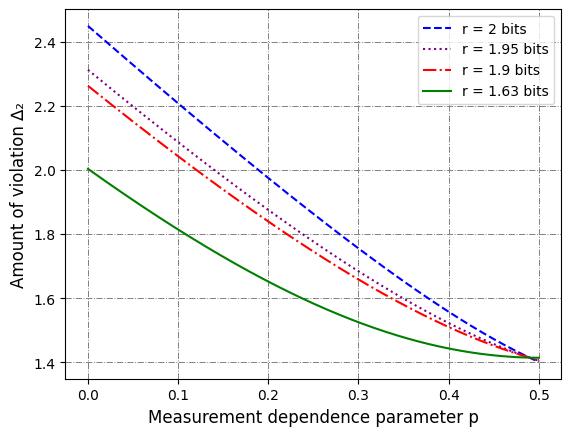}
    \caption{The amount of violation ($\Delta_2$) of the MD-LHV-LHS inequality as a function of the measurement dependence parameter $p$ for different amounts of randomness.}
    \label{fig:randomess2}
\end{figure}

\section{Discussion} \label{sec:discussion}
In this paper, we analyse the implications of relaxing the measurement independence assumption in the bipartite steering scenario. We show that the joint probability distribution obtained by performing local measurements on an MD-LHS assemblage has an MD-LHV model. This implies that the MD-unsteerable states are MD-local. This naturally raises the question of whether all MD-local states are MD-unsteerable, a problem we leave for future research. 

We introduce measurement-dependent weight as a parameter to quantify the measurement dependence for one of the parties in the bipartite steering scenario. In an adversarial scenario, the adversary may prepare the observed assemblage via a convex combination of a steerable assemblage and an MD-unsteerable assemblage. In such a scenario, we present a bound on the measurement-dependent weight such that the observed assemblage is steerable.

For the behavior obtained from joint measurements on the states shared by two parties, we provide an inequality to test for MD-LHV-LHS models. We show that there exist quantum states that violate such an inequality. We also quantify the amount of violation as a function of the measurement dependence parameter of one of the parties such that the obtained behavior contains two bits of randomness. The obtained inequality is expected to have applications in the self-testing of quantum states and other device-independent information processing protocols, like randomness generation and secure communication.

In an experiment, the joint probability distribution describing the experiments' runs may be known by checking the frequencies over multiple runs. The results of this work hold if the runs of the experiment are independent and identically distributed, i.e., the device has no memory. A possible future direction would be to consider the effect of device memory in tests for MD-LHV-LHS models. 

Another possible future direction would be to study the amount of random bits that can be extracted on relaxing the free-will assumption in the tests of genuine multipartite steering and nonlocality~\cite{HR13,HWD22}. 

\begin{acknowledgements}
The authors thank Prerna Rao for the discussion and comments on the manuscript. A.S. thanks Antonio Ac\'{\i}n, Stefan B\"auml, and Arnab Ghorui for the discussions. 
S.D. acknowledges support from the Science and Engineering Research Board, Department of Science and Technology (SERB-DST), Government of India under Grant No.~SRG/2023/000217. S.D. also thanks IIIT Hyderabad for the Faculty Seed Grant.
\end{acknowledgements}

\appendix
\section{The MD-nonlocality scenario}
We consider the Bell scenario where two parties, Alice and Bob, share a quantum state $\rho_{AB}$. Each party can perform one of two measurements available to them. We attribute POVM $\{\Lambda_a^x\}_x$ to Alice and $\{\Lambda_b^y\}_y$ to Bob. Each of these measurements can have two outcomes denoted by $a$ and $b$ for Alice and Bob respectively. The measurement outcome statistics can be described by the probability distribution $\mathbf{P} = \{p(ab|xy)\}$, where $p(ab|xy) = \Tr[\Lambda_a^x \otimes \Lambda_b^y \rho_{AB}]$. The probability distribution $\mathbf{P}$ is also called behavior. 

Let there exist some hidden variable $\lambda$ such that the probability distribution of the outputs conditioned on the inputs can be expressed as
\begin{equation} \label{eq:MDL}
    p(ab|xy) = \int d\lambda~p(ab|xy\lambda)~p(\lambda|xy). 
\end{equation}
The hidden variable $\lambda$ provides an explanation of the observed experimental statistics. An adversary can bias the choice of measurement settings of Alice (or Bob) on the $\lambda$ scale according to
\begin{equation} \label{eq:measurementDep}
    p(x) = \sum_{\lambda} p(x|\lambda) p(\lambda)
\end{equation}
and trick them into believing that they are choosing all the settings with equal probability. As an example, let $\lambda$ take two values $\lambda_1$ and $\lambda_2$ distributed according to $p(\lambda_1) = \sin^2 (\delta_\lambda)$ and $p(\lambda_2) = \cos^2 (\delta_\lambda)$ respectively. In the simplest scenario where $x \in \{1, 2\}$, let the probability of choosing $x$ be distributed according to Table.~\ref{Table:distributionSetting}.  

a. For $\theta_\lambda = 0.3$, $\phi_\lambda = 2.051$, and $\delta_\lambda = 2.447$, it is observed that $p(x) = 0.5$

b. For $\theta_\lambda = 0.7$, $\phi_\lambda = 2.179$, and $\delta_\lambda = 0.96$, it is observed that $p(x) = 0.5$

The above examples show that by properly choosing parameters from Table~\ref{Table:distributionSetting}, an adversary can trick Alice (or Bob) into thinking they have free-will in choosing the measurement settings. For Alice (or Bob) choosing settings with unequal probability, i.e., $p(x) \neq 0.5$, the adversary may adjust the parameters of Table~\ref{Table:distributionSetting} accordingly.  

\begin{table}[t]
    \begin{center}
\begin{tabular}{ m{0.9cm} | m{3cm}  m{3cm}}
  & Distribution 1 $\lambda_1$ & Distribution 2 $\lambda_2$ \\
  \hline
  $p(1|\lambda)$ & $\cos^2 (\theta_\lambda)$ & $\cos^2 (\phi_\lambda)$  \\
  $p(2|\lambda)$ & $\sin^2 (\theta_\lambda)$ & $\sin^2 (\phi_\lambda)$  \\
\end{tabular} 
\caption{Probability distribution for choice of settings by Alice (or Bob) based on the hidden variable $\lambda$. An adversary can use this distribution, and by suitably choosing the parameters of the table the adversary can trick Alice (or Bob) into thinking they have free-will
in choosing the measurement settings.}
\label{Table:distributionSetting}
\end{center}
\end{table}

In a scenario where there is a bias in the choice of measurement settings in the $\lambda$ scale which Alice and Bob are unaware of, the following constraints can be imposed on the conditional joint probability distribution:

a. Signal locality imposes the factorisability constraint on the conditional joint probability distribution. 
\begin{equation}
        p(ab|xy\lambda) = p(a|x\lambda) p(b|y\lambda)
\end{equation}  

b. The measurement independence assumption requiring that $\lambda$ does not contain any information about $x$ and $y$, which implies
\begin{eqnarray}
        &&p(\lambda|xy) = p(\lambda) \\
        \text{or equivalently,} &&p(xy|\lambda) = p(xy).
\end{eqnarray}

\section{Measurement-dependent steering weight} \label{app:steering-weight}
It follows from Eq.~\eqref{eq:steringMixture} that 
\begin{eqnarray}
    \gamma_{a|x}^{\text{st}} = \frac{1}{1-\eta^{a|x}}\bigg[\sigma_{a|x} - \eta^{a|x}~\sigma_{a|x}^{\text{MD-LHS}}\bigg] \geq 0,
\end{eqnarray}
which is equivalent to stating $\sigma_{a|x} - \eta^{a|x}~\sigma_{a|x}^{\text{MD-LHS}} \geq 0$. Expressing $\sigma_{a|x}^{\text{LHS}}$ as in Eq.~\eqref{eq:steeringDef}, we obtain
\begin{eqnarray}
    \sigma_{a|x} \geq \eta^{a|x} \int d\lambda p(\lambda|x) p(a|x\lambda) \rho_{\lambda|x}.
\end{eqnarray}
Noting that $\sigma_{a|x} = p(a|x) \rho_{a|x}$ and taking trace on both sides, we obtain
\begin{eqnarray}
    p(a|x) \geq \eta^{a|x} \int d\lambda p(\lambda|x) p(a|x\lambda). \label{eq:prop1-1}
\end{eqnarray}
It follows from Eq.~\eqref{eq:prop1-1} that
\begin{eqnarray}
    && p(+|x_1) - p(+|x_2) \geq \int d\lambda~\bigg[ \eta^{+|x_1} p(\lambda|x_1) p(+|x_1\lambda) \nonumber \\
    && \qquad \qquad \qquad \qquad \qquad - \eta^{+|x_2} p(\lambda|x_2) p(+|x_2\lambda)\bigg] \label{eq:prop1-2}\\
    &&\text{and} \nonumber \\
    && p(+|x_1) - p(+|x_2) \leq \int d\lambda~\bigg[ \eta^{-|x_1} p(\lambda|x_1) p(-|x_1\lambda) \nonumber \\
    && \qquad \qquad \qquad \qquad \qquad - \eta^{-|x_2} p(\lambda|x_2) p(-|x_2\lambda)\bigg], \nonumber \\ \label{eq:prop1-3}
\end{eqnarray}
where in the last inequality, we have taken $p(+|x_1) = 1 - p(-|x_1)$ and $p(+|x_2) = 1 - p(-|x_2)$. It then follows from Eq.~\eqref{eq:prop1-2} and Eq.~\eqref{eq:prop1-3} that 
\begin{eqnarray}
    &&\int d\lambda \bigg[ p(\lambda|x_1) \bigg\{ \eta^{+|x_1} p(+|x_1 \lambda) - \eta^{-|x_1} p(-|x_1 \lambda) \bigg\} \nonumber \\
    &&-p(\lambda|x_2) \bigg\{ \eta^{+|x_2} p(+|x_2 \lambda) - \eta^{-|x_2} p(-|x_2 \lambda) \bigg\} \bigg] \leq 0. \nonumber 
\end{eqnarray}
With the observation that $p(-|x_j\lambda) = 1 - p(+|x_j\lambda)$ and rearranging we have
\begin{eqnarray}
    &&\int d\lambda p(\lambda|x_1) p(+|x_1\lambda) \big(\eta^{+|x_1} + \eta^{-|x_1}\big) \nonumber \\
    && \qquad - \int d\lambda p(\lambda|x_2) p(+|x_2\lambda) \big(\eta^{+|x_2} + \eta^{-|x_2}\big) \nonumber \\
    && \qquad \qquad - \eta^{-|x_1} \int d\lambda \bigg( p(\lambda|x_1) \nonumber \\
    && \qquad \qquad \qquad \qquad \qquad - \frac{\eta^{-|x_2}}{\eta^{-|x_1}} p(\lambda|x_2) \bigg) \leq 0.
\end{eqnarray}
Invoking Eq.~\eqref{eq:md-para} and properties of conditional joint probability in the above inequality, we obtain
\begin{eqnarray}
    &&\int d\lambda p(+\lambda|x_1) \big(\eta^{+|x_1} + \eta^{-|x_1}\big) \nonumber \\
    && \qquad - \int d\lambda p(+\lambda|x_2) \big(\eta^{+|x_2} + \eta^{-|x_2}\big) - \eta^{-|x_1} \mathscr{W}  \leq 0. \nonumber \\
\end{eqnarray}
Next, on simplifying, we obtain
\begin{eqnarray}
    &&\mathscr{W} \leq \frac{1}{\eta^{-|x_1}} \bigg[ p(+|x_2)(\eta^{+|x_2} + \eta^{-|x_2}) \nonumber \\
    &&\qquad \qquad \qquad \quad - p(+|x_1)(\eta^{+|x_1} + \eta^{-|x_1})\bigg].
\end{eqnarray}

\section{MD-LHV-LHS Inequality} \label{app:MD-LHV-LHS}
The set of correlations in Eq.~\eqref{eq:steeringInequality1} forms a convex set and can be expressed in terms of extremal points as
\begin{eqnarray}
    p(ab|xy) = &&\sum_{\chi} \int d\xi~\delta_{a,\tau(x,\chi)}\Tr[\Lambda_b^y \rho_\lambda]. \nonumber \\
    &&\qquad \qquad \bigg[\sum_{\lambda} p(\chi\xi|\lambda x y
    ) p(\lambda|xy)\bigg] \\
    = && \sum_{\chi} \int d\xi~\delta_{a,\tau(x,\chi)}\Tr[\Lambda_b^y \rho_\lambda]~p(\chi\xi|xy), \nonumber \\
\end{eqnarray}
where $\Lambda_{b}^y$ is the projector of Bob that produces outcome $b$ for measurement $y$ on the state $\rho_\lambda = \op{\psi_\xi}$. The parameter $\xi$ determines the state $\ket{\psi_\xi}$ for Bob and $\chi$ determines the value of $a$ from the function $\tau(x,\chi)$. 

For an MD-LHV-LHS model, the correlations of Alice and Bob can be reproduced if there exists a probability distribution $p(\chi\xi|xy)$ such that 
\begin{eqnarray}
    \langle xy \rangle &&= \sum_{\chi} \int d\xi p(\chi\xi|xy) \bigg[p_+^x(\chi) p_+^y(\xi) + p_-^x(\chi) p_-^y(\xi) \nonumber \\
    && \qquad - p_+^x(\chi) p_-^y(\xi) - p_-^x(\chi) p_+^y(\xi) \bigg] \\
    && = \sum_{\chi} \int d\xi p(\chi\xi|xy) [2p^x_+(\chi) - 1] [2p^y_+(\xi) -1]\nonumber\\ 
    && = \frac{1}{p(xy)} \sum_{\chi} \int d\xi p(\chi\xi) \bigg\{p(xy|\chi\xi) [2p^x_+(\chi) - 1] \nonumber \\\
    && \qquad \qquad \qquad \qquad \qquad \qquad [2p^y_+(\xi) -1]\bigg\}, \label{eq:expecXY}
\end{eqnarray}
where $p_a^x(\chi) = p(a|x\chi)$ and $p_b^y(\xi) = p(b|y\psi_\xi)$. In the second equality we have used $p^x_+(\chi) + p^x_-(\chi) = 1$ and $p^x_+(\xi) + p^x_-(\xi) = 1$ while in the third equality we have used $p(\chi\xi|xy) = p(\chi\xi)~p(xy|\chi\xi)/p(xy)$. 

Let Alice have four extreme values of $\chi \in \{\chi_1,\chi_2,\chi_3,\chi_4\}$ corresponding to $p_+^{x_1} = p_+^{x_2} = 1$, $p_+^{x_1} = p_+^{x_2} = 0$, $p_+^{x_1} = 1-p_+^{x_2} = 1$, $p_+^{x_1} = 1 - p_+^{x_2} = 0$. Let $y_1 \text{ and }y_2$ be projective measurements written as $y_1 = 2 \prod_1^{y_1} - \mathbbm{1}$, $\prod_1^{y_1}$ being the projector on the $+1$ eigenstate of $y_1$ and similarly for $y_2$. Following the parametrisation from \cite{cavalcanti2015analog} let us define $\mu = \Tr[\prod_1^{y_1} \prod_1^{y_2}]$. Then $(p_1^{y_1}(\xi),p_1^{y_2}(\xi))$ form an ellipse parametrised as
\begin{eqnarray}
    &&2 p_1^{y_1}(\xi) - 1 = \cos(\xi + \beta) \\
    \text{and}~&&2 p_1^{y_2}(\xi) - 1 = \cos(\xi - \beta), 
\end{eqnarray}
with $\beta = \arctan(\frac{\sqrt{1-\mu}}{\sqrt{\mu}}) \in [0,\pi/2]$. It was shown in \cite{girdhar2016all} that if $y_1 \text{ and } y_2$ are dichotomic then $(p_1^{y_1}(\xi),p_1^{y_2}(\xi))$ form an ellipse. For each value of $\chi$ and $\xi$, the correlations in Eq~\eqref{eq:expecXY} can be expressed as shown in Table~\ref{Table:correlationLHVLHS}.
\begin{table}[t]
    \begin{center}
\begin{tabular}{ m{0.9cm} | m{1.5cm} | m{2cm} | m{2cm} | m{2cm}}
  & $\chi = \chi_1$ & $\chi = \chi_2$ & $\chi = \chi_3$ & $\chi = \chi_4$ \\
  \hline
$\langle x_1 y_1 \rangle$ & $p(11|\chi_1\xi) \newline \cos(\xi+\beta)$ & $(-1) p(11|\chi_2\xi) \newline \cos(\xi+\beta)$ & $p(11|\chi_3\xi) \newline \cos(\xi+\beta)$ & $(-1) p(11|\chi_4\xi) \newline \cos(\xi+\beta)$ \\
\hline
$\langle x_1 y_2 \rangle$ & $p(12|\chi_1\xi) \newline \cos(\xi-\beta)$ & $(-1) p(12|\chi_2\xi) \newline \cos(\xi-\beta)$ & $p(12|\chi_3\xi) \newline \cos(\xi-\beta)$ & $(-1) p(12|\chi_4\xi) \newline \cos(\xi-\beta)$ \\
\hline
$\langle x_2 y_1 \rangle$ & $p(21|\chi_1\xi) \newline \cos(\xi+\beta)$ & $(-1) p(21|\chi_2\xi) \newline \cos(\xi+\beta)$ & $(-1) p(21|\chi_2\xi) \newline \cos(\xi+\beta)$ & $p(21|\chi_4\xi) \newline \cos(\xi+\beta)$ \\
\hline
$\langle x_2 y_2 \rangle$ & $p(22|\chi_1\xi) \newline \cos(\xi-\beta)$ & $(-1) p(22|\chi_2\xi) \newline \cos(\xi-\beta)$ & $(-1) p(22|\chi_3\xi) \newline \cos(\xi-\beta)$ & $p(22|\chi_4\xi) \newline \cos(\xi-\beta)$ \\
\end{tabular} 
\caption{In the above table $p(ij|\chi_k\xi)$ denotes the probability $p(x_iy_j|\chi_k\xi)$.}
\label{Table:correlationLHVLHS}
\end{center}
\end{table}
These correlations have an MD-LHV-LHS model if they can be written as a convex combination of elements of Table~\ref{Table:correlationLHVLHS}. 

Assuming that Bob has a fully characterised system, we have
\begin{eqnarray}
    p(xy|\chi\xi) = p(x|\chi \xi) p(y|x\chi\xi) = \frac{1}{2}~p(x|\chi \xi). \label{eq:AliceMD}
\end{eqnarray}
In the second equality, we assume no measurement dependence for Bob. Inserting Eq.~\eqref{eq:AliceMD} in Table~\ref{Table:correlationLHVLHS} we obtain the correlations presented in Table~\ref{Table:correlationLHVLHSBob}.
\begin{table}[t]
    \begin{center}
\begin{tabular}{ m{0.9cm} | m{1.5cm} | m{1.9cm} | m{1.9cm} | m{1.9cm}}
  & $\chi = \chi_1$ & $\chi = \chi_2$ & $\chi = \chi_3$ & $\chi = \chi_4$ \\
  \hline
$\langle x_1 y_1 \rangle$ & $2p(1|\chi_1 \xi) \newline \cos(\xi+\beta)$ & $(-2) p(1|\chi_2 \xi) \newline \cos(\xi+\beta)$ & $2p(1|\chi_3 \xi) \newline \cos(\xi+\beta)$ & $(-2) p(1|\chi_4 \xi) \newline \cos(\xi+\beta)$ \\
\hline
$\langle x_1 y_2 \rangle$ & $2p(1|\chi_1 \xi) \newline \cos(\xi-\beta)$ & $(-2) p(1|\chi_2 \xi) \newline \cos(\xi-\beta)$ & $2p(1|\chi_3 \xi) \newline \cos(\xi-\beta)$ & $(-2) p(1|\chi_4 \xi) \newline \cos(\xi-\beta)$ \\
\hline
$\langle x_2 y_1 \rangle$ & $2p(2|\chi_1 \xi) \newline \cos(\xi+\beta)$ & $(-2) p(2|\chi_2 \xi) \newline \cos(\xi+\beta)$ & $(-2) p(2|\chi_3 \xi) \newline \cos(\xi+\beta)$ & $2p(2|\chi_4 \xi) \newline \cos(\xi+\beta)$ \\
\hline
$\langle x_2 y_2 \rangle$ & $2p(2|\chi_1 \xi) \newline \cos(\xi-\beta)$ & $(-2) p(2|\chi_2 \xi) \newline \cos(\xi-\beta)$ & $(-2) p(2|\chi_3 \xi) \newline \cos(\xi-\beta)$ & $2p(1|\chi_4 \xi) \newline \cos(\xi-\beta)$ \\
\end{tabular} 
\caption{In the above table $p(i|\chi_j \xi)$ denotes the probability $p(x_i|\chi_j \xi)$.}
\label{Table:correlationLHVLHSBob}
\end{center}
\end{table}
Let $C_j$ be the convex hull of the column with $\chi = \chi_j$ and let $p(1|\chi_1 \xi) = p(1|\chi_2 \xi) = p(1|\chi_3 \xi) = p(1|\chi_4 \xi) = p_1$ and $p(2|\chi_1 \xi) = p(2|\chi_2 \xi) = p(2|\chi_3 \xi) = p(2|\chi_4 \xi) = p_2$. In the basis $\textbf{e}_1 = (1,0,0,0)^T$, $\textbf{e}_2 = (0,1,0,0)^T$, $\textbf{e}_3 = (0,0,1,0)^T$, $\textbf{e}_4 = (0,0,0,1)^T$, the vectors forming the boundary of $C_j$ are given by $\textbf{C}_1 = 2\cos(\xi + \beta)(p_1~\textbf{e}_1 + p_2~\textbf{e}_3) + 2\cos(\xi - \beta)(p_1\textbf{e}_2 + p_2\textbf{e}_4)$, $\textbf{C}_2 = - \textbf{C}_1$, $\textbf{C}_3 = 2\cos(\xi + \beta)(p_1~\textbf{e}_1 - p_2~\textbf{e}_3) + 2\cos(\xi - \beta)(p_1\textbf{e}_2 - p_2\textbf{e}_4)$ and $\textbf{C}_4 = -\textbf{C}_3$. The curve with the parametrisation $u_1 = 2 p_1 cos(\xi + \beta)$, $u_2 = 2 p_1 cos(\xi - \beta)$, $u_3 = 2 p_2 cos(\xi + \beta)$ and $u_4 = 2 p_2 cos(\xi - \beta)$ can be expressed as 
\begin{eqnarray}
    &&\frac{1}{16 p_1^2 p_2^2 \sin^2(2\beta)}\bigg((p_2u_1 + p_1u_3)^2 + (p_2u_2 + p_1u_4)^2 \nonumber \\
    &&- 2(p_2u_1 + p_1u_3)(p_2u_2 + p_1u_4)\bigg) = 1. 
\end{eqnarray}
Also, the curve with the parametrisation $v_1 = 2 p_1 cos(\xi + \beta)$, $v_2 = 2 p_1 cos(\xi - \beta)$, $v_3 = -2 p_2 cos(\xi + \beta)$ and $v_4 = -2 p_2 cos(\xi - \beta)$ can be expressed as 
\begin{eqnarray}
    &&\frac{1}{16 p_1^2 p_2^2 \sin^2(2\beta)}\bigg((p_2v_1 - p_1v_3)^2 + (p_2v_2 - p_1v_4)^2 \nonumber \\
    &&- 2(p_2v_1 - p_1v_3)(p_2v_2 - p_1v_4)\bigg) = 1. 
\end{eqnarray}
A vector that has an MD-LHV-LHS model can be written in the form of a superposition as $\textbf{C} = \textbf{C}_1 + \textbf{C}_2$. Now if $\textbf{C}$ lies in the convex hull of $\textbf{C}_1$ and $\textbf{C}_2$ then $\textbf{C} = \nu_1 \textbf{w}_1 + \nu_2 \textbf{w}_2$ where $\textbf{w}_1 \subseteq \textbf{C}_1$, $\textbf{w}_2 \subseteq \textbf{C}_2$, and $\nu_1,\nu_2 \in (0,1)$. It then follows that $\textbf{w}_1 = \textbf{C}_1/\nu_1$ and $\textbf{w}_2 = \textbf{C}_2/\nu_2$. If we denote $f_1 \leq r_1$ and $f_2 \leq r_2$ as the boundaries of $\textbf{C}_1$ and $\textbf{C}_2$ then $f_1 \leq \nu_1^2~r_1$ and $f_2 \leq \nu_2^2~r_2$ which implies
\begin{eqnarray}
    \sqrt{f_1} + \sqrt{f_2} && \leq \nu_1 \sqrt{r_1} + \nu_2 \sqrt{r_2} \nonumber \\
    && \leq \max \{\sqrt{r_1},\sqrt{r_2}\},
\end{eqnarray}
where the last inequality follows from observing $\nu_1 + \nu_2 = 1$. In terms of the expectation value of the operators, this is equivalent to saying 
\begin{equation}
    \sqrt{\alpha_1} + \sqrt{\alpha_2} \leq 4p_1p_2 \sin 2\beta, \label{eq:steeringInequalityBob}
\end{equation}
where we have
\begin{eqnarray}
    &&\alpha_1 = (p_2\langle x_1y_1\rangle + p_1 \langle x_2y_1\rangle)^2 + (p_2\langle x_1y_2\rangle + p_1 \langle x_2y_2\rangle)^2 \nonumber \\
    &&- 2 (p_2\langle x_1y_1\rangle + p_1 \langle x_2y_1\rangle) \nonumber \\
    &&\qquad \qquad \qquad (p_2\langle x_1y_2\rangle + p_1 \langle x_2y_2\rangle)\cos(2\beta). \\
    &&\alpha_2 = (p_2\langle x_1y_1\rangle - p_1 \langle x_2y_1\rangle)^2 + (p_2\langle x_1y_2\rangle - p_1 \langle x_2y_2\rangle)^2 \nonumber \\
    &&- 2 (p_2\langle x_1y_1\rangle - p_1 \langle x_2y_1\rangle) \nonumber \\
    &&\qquad \qquad \qquad (p_2\langle x_1y_2\rangle - p_1 \langle x_2y_2\rangle)\cos(2\beta). 
\end{eqnarray}
Now choosing $\beta = \pi/4$ and observing $p_1 = 1-p$ and $p_2 = p$, we have from Eq.~\eqref{eq:steeringInequalityBob} the following inequality
\begin{equation}
    \sqrt{\alpha'_1} + \sqrt{\alpha'_2} \leq 4p(1-p),  \label{eq:steeringInequalityBob123}
\end{equation}
where 
\begin{eqnarray}
    &&\alpha'_1 = (p\langle x_1y_1\rangle + (1-p) \langle x_2y_1\rangle)^2 + (p\langle x_1y_2\rangle \nonumber \\
    && \qquad \qquad \qquad \qquad \qquad  + (1-p) \langle x_2y_2\rangle)^2 \label{eq:alpha1}\\
    &&\alpha'_2 = (p\langle x_1y_1\rangle - (1-p) \langle x_2y_1\rangle)^2 + (p\langle x_1y_2\rangle \nonumber \\
    &&\qquad \qquad \qquad  \qquad \qquad - (1-p) \langle x_2y_2\rangle)^2. \label{eq:alpha2}
\end{eqnarray}

\bibliographystyle{apsrev4-1}
\bibliography{paper}{}
\end{document}